\documentclass[pra,showpacs,twocolumn,amsmath]{revtex4}
\usepackage{pifont}
\usepackage{graphicx,epsfig,subfigure,dsfont,amssymb,amsmath,amsthm,amsfonts,amsbsy,mathrsfs,amscd}

\DeclareMathOperator{\tr}{tr}

\newtheorem{theorem}{Theorem}

\newtheorem{lemma}{Lemma}

\def\la{\langle}
\def\ra{\rangle}

\input amssym.def

\begin{document}
\title{Remote Creation of Quantum Coherence}

\smallskip
\author{Teng Ma$^1$ }
\author{Ming-Jing Zhao$^2$}
\author{Shao-Ming Fei$^{3,4}$}
\author{Gui-Lu Long$^1$}

\affiliation{$^1$State Key Laboratory of Low-Dimensional Quantum Physics and Department of Physics, Tsinghua University, Beijing 100084, China\\
$^2$School of Science,
Beijing Information Science and Technology University, Beijing, 100192, China\\
$^3$School of Mathematical Sciences, Capital Normal
University, Beijing
100048, China\\
$^4$Max-Planck-Institute for Mathematics in the Sciences, 04103
Leipzig, Germany}


\begin{abstract}
We study remote creation of coherence (RCC) for a quantum system A with the help of quantum operations on another system B and one way classical communication. We show that all the non-incoherent quantum states are useful for RCC and all the incoherent-quantum states are not. The necessary and sufficient conditions of RCC for the quantum operations on system B are presented for pure states. The upper bound of average RCC is derived, giving a relation among the entanglement (concurrence), the RCC of the given quantum state and the RCC of the corresponding maximally entangled state. Moreover, for two-qubit systems we find a simple factorization law for the average remote created coherence.
\end{abstract}

\maketitle

\section{Introduction}

Quantum coherence is one of the fundamental features which distinguish quantum world from classical realm.
It is the origin for extensive quantum phenomena such as interference, laser, superconductivity \cite{london} and superfluidity \cite{migdal}.
It is an important subject related to quantum mechanics, from quantum optics \cite{mandel}, solid state physics \cite{london,migdal},
thermodynamics \cite{johan aberg}, to quantum biology \cite{engel}.
Coherence, together with quantum correlations like quantum entanglement \cite{horodecki}, quantum discord \cite{modi}, are
crucial ingredients in quantum  computation and  information tasks \cite{nielsen}.
Coherence  shines its quantum merits in quantum metrology \cite{giovannetti,dobranski}, quantum key distribution \cite{toshihiko sasaki}, entanglement creation \cite{asboth,streltsov}, etc.

Unlike quantum entanglement and other quantum correlations, coherence, regarded as a physical resource \cite{johan aberg,toshihiko sasaki,t baumgratz},
has been just investigated very recently in establishing the framework of quantifying coherence in the language of quantum information theory \cite{bromley,streltsov,e chitambar,yy,yuan,bera,girolami,mani,1,2,3,4,5,6,7,8,9,10}.
Due to its fundamental role in quantum physics and quantum information theory, it is still necessary to understand how coherence works in information processing and investigate the relations between coherence and quantum correlations.

In this paper, we consider the creation of coherence of a quantum system A, which initially has zero coherence, with the help of quantum channels on another quantum system B.
Such creation of coherence depends on the correlations between A and B, as well as the quantum operations on system B. We  establish an explicit relation among the
creation of the coherence on system A, the quantum entanglement (concurrence) between A and B,
and the quantum operations.

Remote creation of coherence (RCC) can be illustrated by a simple example. Consider  a two-qubit system AB, which is
initially in the maximally entangled state $(|00\rangle + |11\rangle)/\sqrt{2}$. If B undergoes a projective measurement under  basis $\{|\beta_0\ra, |\beta_1\ra\}$, and tells measurement outcome, for example $|\beta_0\ra$, to A,
the system A's final state would be a superposition of $|0\ra$ and $|1\ra$ with some probability,
if the basis $|\beta_0\ra$ is neither $|0\ra$ nor $|1\ra$. The same analysis also holds for measurement  outcome $|\beta_1\ra$. Namely, system A can gain an averaged coherence over all the outcomes.

For general states,  we first investigate what kind of states can be used to create remote coherence and present a necessary and sufficient condition. Then for pure states, we give a necessary and sufficient condition that the operations must satisfy for nonzero RCC.
And finally, we give upper bounds of RCC and investigate the relations between coherence and entanglement in our scenario.

Since a quantum state's coherence depends on the reference basis, throughout our paper, we fix the system A's reference basis to be the computational basis.
A well defined and mostly used coherence measure is the $l_1$ norm coherence $C_{l_1}$  \cite{t baumgratz}.
The $l_1$ norm coherence of a quantum state is defined as the sum of all off-diagonal elements of the state's density matrix under the reference basis, i.e.,
\begin{equation}\label{cl1}
C_{l_1}(\rho)=\sum_{i\neq j} |\rho_{ij}|,
\end{equation}
where $|\rho_{ij}|$ is the absolute value of $\rho_{ij}$.
In the following, we use the $l_1$ norm coherence when we discuss the upper bound of RCC and its relation to entanglement.

\section{Conditions for creating RCC}

Let ${\rho^{AB}}$ be a bipartite quantum state and $\$(\cdot)=\sum_n F_n(\cdot){F_n}^\dagger$  a quantum operation acting on the subsystem B.
Let ${\rho^A}'=\tr_B[(\mathbb{I}\otimes \$) {\rho^{AB}}]/\tr[(\mathbb{I}\otimes \$) {\rho^{AB}}]$ be the reduced state of the system A
after the operation. Concerning the quantum operations used for remote creation of coherence with a general quantum state, we have the following theorem.
\begin{theorem}\label{th1}
 Given a bipartite quantum state  ${\rho^{AB}}$, for any quantum operation $\$(\cdot)=\sum_n F_n(\cdot){F_n}^\dagger$ acting on the subsystem B,
the coherence of the final subsystem A $C({\rho^A}')=0 $ if and only if ${\rho^{AB}}$ is an incoherent-quantum state
${\rho^{AB}}=\sum_i p_i \sum_k q^i_k|k\ra \la k|\otimes \rho_i^B$.
\end{theorem}

\begin{proof}
Let ${\rho^{AB}}=\sum_{ijkl}p_{ij,kl}|i\ra \la k|\otimes |j\ra \la l|$ be the quantum state  of system AB under the computational basis,
then the marginal state  $\rho^A=\tr_B{\rho^{AB}}=\sum_{ikj} p_{ij,kj}|i\ra \la k|$.
After the quantum operation $\$$ acting on B, $\rho^{AB}$ becomes,
\begin{equation}\nonumber
\begin{aligned}
{\rho^{AB}}'=&(\mathbb{I}\otimes \$) {\rho^{AB}}/p' \\
=&\frac{1}{p'} \sum_{ijkl} p_{ij,kl}|i\ra \la k|\otimes \sum_n F_n |j\ra \la l| {F_n}^\dagger,
\end{aligned}
\end{equation}
where $p'=\tr [(\mathbb{I}\otimes \$) {\rho^{AB}}]$ is the probability of getting the state ${\rho^{AB}}'$. Tracing over the system B, we get the final state of system A,
\begin{equation}\label{mixrhoap}
 {\rho^A}'=1/p' \sum_{ik} (\sum_{jl}p_{ij,kl} N_{lj}) |i\ra \la k|,
\end{equation}
where $N_{lj}=\la l|N|j\ra$ and $N=\sum_n {F_n}^\dagger {F_n}$.

``$\Rightarrow$''.~From equations (\ref{cl1}) and (\ref{mixrhoap}), for any operation $\$$, $C({\rho^A}')=0$ means that for any $N$ and $i\neq k$, $\sum_{jl} p_{ij,kl} N_{lj}=0$, i.e.,
\begin{equation}\label{mix1}
\sum_j p_{ij,kj} N_{jj}+\sum_{j< l}p_{ij,kl} N_{lj}+\sum_{j< l} p_{il,kj} N_{jl}=0.
\end{equation}
 For the arbitrary Hermitian operator $N\leq \mathbb{I}$ (the operation $\$$ is aribitray), $N_{jl}={N^*_{lj}}$ and the diagonal entries
of $N$ are arbitrary real numbers which are independent of the off diagonal entries.
Thus from (\ref{mix1}), we get $p_{ij,kj}=0$, $\forall j, i\neq k$. On the other hand, for all $i\neq k$,
$$ \sum_{j< l}p_{ij,kl} N_{lj}+ p_{il,kj} {N^*_{lj}}=0.$$
Set $N_{lj}=a_{lj}+b_{lj}I$, where $I=\sqrt{-1}$ is the unit imaginary number. Substituting it into the above equation, we get
\begin{equation}\label{mix2}
\begin{aligned}
%&\sum_{j<l} p_{ij,kl}(a_{lj}+b_{lj}I)+\sum_{j<l} p_{il,kj}(a_{lj}-b_{lj}I)=0\\
\sum_{j<l}[(p_{ij,kl}+ p_{il,kj})a_{lj}+(p_{ij,kl}- p_{il,kj})b_{lj}I]=0.
\end{aligned}
\end{equation}
Since $N_{lj}$ $(j<l)$ is arbitrary, $a_{lj}$ and $b_{lj}$ are all independent. Then equation (\ref{mix2}) implies that for all $i\neq k, j<l$,
$p_{ij,kl}+ p_{il,kj}=p_{ij,kl}-p_{il,kj}=0$, i.e., $p_{ij,kl}=p_{il,kj}=0$, $\forall~  i\neq k,j<l$. Therefore $p_{ij,kl}=0$, $\forall j,l,i\neq k$.
Hence the initial state ${\rho^{AB}}$ becomes ${\rho^{AB}}=\sum_{ijl}p_{ij,il} |i\ra \la i|\otimes  |j\ra \la l|$ which actually is an incoherent-quantum state \cite{streltsov}.

``$\Leftarrow$''.~If ${\rho^{AB}}$ is an incoherent-quantum state, ${\rho^{AB}}=\sum_i p_i \sum_k q^i_k|k\ra \la k|\otimes \rho_i^B$.
It is easy to check that, for any operation $\$$ acting on system B, the final state of the system A,
${\rho^A}'=\tr_B[(\mathbb{I}\otimes \$) {\rho^{AB}}]/\tr[(\mathbb{I}\otimes \$) {\rho^{AB}}]$, has zero coherence.
\end{proof}

From the above proof one can see that $\forall~\$, C({\rho^A}')=0$ also means that the initial coherence of system A is also zero, $C(\rho^A)=0$.
The above theorem implies that any non-incoherent quantum state can be used for RCC under certain operations.
Note that the incoherent-quantum state actually is the classical-quantum correlated state with the fixed reference basis \cite{streltsov}.
Thus the states, that can be used to create coherence, are not limited within entangled states. According to theorem \ref{th1}
a non-incoherent-quantum separable state, with the system A having zero coherence, can also be used for RCC.

Interestingly, the condition in theorem \ref{th1} for creating nonzero RCC is the same as the distillable coherence of collaboration in \cite{e chitambar}.
However, it should be noticed that our scenario is different from the asymptotic scenario in \cite{e chitambar}. In \cite{e chitambar}, they study the maximal distilled coherence of collaboration under local quantum-incoherent operations and bilateral  classical communications for the asymptotic case.
While our work investigates the system A's average coherence after system B going  through a \emph{certain} quantum channel (see the rest of our paper),
and the RCC in our scenario only requires one way classical communication and does not involve the maximal process \cite{e chitambar}.

The next natural question is, for a non-incoherent quantum state, what is the exact form of the operation acting on system B
for the creation of coherence? Obviously for operation $\$=\sum_n F_n(\cdot) F_n^\dagger$ acting on system B such that $N=\sum_n F_n^\dagger F_n=\mathbb{I}/q$, where $q$
is an arbitrary real number bigger than 1, the coherence can not be created, which can be seen
by substituting $N=\mathbb{I}/q$ into equation (\ref{mixrhoap}). One can check that the important quantum operations including depolarizing operations, phase flip operations, bit flip operations and bit-phase flip operations all belong  to this form.
Moreover, all the trace preserving quantum channels will not create the coherence remotely,
and measurements on the system B are necessary in order to create nonzero coherence.
We give the necessary and sufficient condition for RCC in the following.
Let ${|\psi^{AB}\ra}$ be a pure bipartite entangled quantum state with zero coherence of the system A. We have

\begin{theorem}\label{th2}
 After a quantum operation $\$(\cdot)=\sum_n F_n(\cdot){F_n}^\dagger$ acts  on the system B, the system A, with initial coherence being zero, gains  coherence if and only if
there is a computational basis $|i\ra$ with $[N, (\la i| \otimes \mathbb{I})|\psi\ra^{AB} \la \psi|(|i\ra \otimes \mathbb{I}\ra)]\neq 0$,
where $|\psi\ra^{AB}$ is the initial state of AB, $N=\sum_n {F_n}^\dagger F_n$ and $[a,b]=ab-ba$ is the Lie bracket.
\end{theorem}

\begin{proof}
 With the local computational basis, the state of system AB can be expressed as $|\psi^{AB}\ra=\sum_{ij} w_{ij} |i\ra |j\ra$, with $\sum_{ij}|w_{ij}|^2=1$ the normalization condition.
One has $\rho^A=\tr_B|\psi\ra ^{AB}\la \psi|=\sum_{ijk}w_{ij}w_{kj}^*|i\ra \la k|$. Respect to $C(\rho^A)=0$, the rows of the coefficient matrix $W=(w_{ij})$ are mutually orthogonal.
Then the singular value decomposition of $W$ has a simple form, $W=DV$, where $D=(\sqrt{\omega_i})$ is a diagonal matrix with nonzero singular values
$\sqrt{\omega_i}$ (zero singular values are trivial for our proof), $V$ is a unitary matrix. Thus we get the Schmidt decomposition,
\begin{equation}\label{psiab}
|\psi^{AB}\ra=\sum_{ij}\sqrt{\omega_i} V_{ij}|i\ra |j\ra=\sum_{i}\sqrt{\omega_i} |i\ra |\beta_i\ra,
 \end{equation}
where $|\beta_i\ra=\sum_j V_{ij}|j\ra$, $\sum_i \omega_i=1$. We also have $(\la i| \otimes \mathbb{I})|\psi^{AB}\ra=\sqrt{\omega_i}|\beta_i\ra$.

After the local operation acting on the system B, the final state of the system A has the form,
\begin{equation}
\begin{aligned}\label{purerhoap}
{\rho^{A}}'=&\tr_B [(\mathbb{I}\otimes \$) |\psi\ra ^{AB} \la \psi|]/p' \\
=&\tr_B [(\mathbb{I}\otimes \$)\sum _{ij} \sqrt{\omega_i\omega_j} |i\ra \la j|\otimes |\beta_i\ra \la \beta_j|]/p' \\
=&\sum_{ij}  \sqrt{\omega_i\omega_j} |i\ra \la j| \tr [\sum_n F_n(|\beta_i\ra \la \beta_j|)F_n^\dagger]/p'\\
=&\sum_{ij}  \sqrt{\omega_i\omega_j} N_{ji} |i\ra \la j|/p',
\end{aligned}
\end{equation}
where  $N=\sum_n F_n^\dagger F_n$, $N_{ji}= \la\beta_j|N|\beta_i\ra$ and $p'=\tr[(\mathbb{I}\otimes \$)|\psi\ra ^{AB} \la \psi|]$, the probability of getting the state ${\rho^A}'$.
Since $\sqrt{\omega_i}\neq 0$ $\forall~i$,
all the off diagonal entries  of ${\rho^A}'$ vanish if and only if $N_{ji}=0, ~\forall ~j\neq i$, which means that
$|\beta_i\ra$s are just the  eigenvectors of $N$. Equivalently, $[N, \omega_i|\beta_i\ra \la \beta_i|]=0$, $\forall~i$.
Thus $C({\rho^A}')\neq 0$ if and only if there is a computational basis $|i\ra$ such that
$[N, \omega_i|\beta_i\ra \la \beta_i|]\neq 0$, where $\sqrt{\omega_i}|\beta_i\ra=(\la i| \otimes \mathbb{I})|\psi^{AB}\ra$.
\end{proof}

For a pure bipartite entangled state $|\psi^{AB}\ra$ with A's initial coherence being zero, we have presented an operational way
to determine whether $\$(\cdot)=\sum_n F_n(\cdot){F_n}^\dagger$ on system B can create system A's coherence.
In fact, from the proof we have also given the explicit form of the quantum operation that can not create  coherence  of system A.
Such operation $N$ has the form: $N=\sum_i n_i |\beta_i\ra \la \beta_i|$, where $n_i$s are arbitrary real numbers in $[0,1]$,
$|\beta_i\ra=(\la i| \otimes \mathbb{I})|\psi^{AB}\ra/\sqrt{ \la \psi^{AB}|(|i\ra \otimes \mathbb{I}\ra)(\la i| \otimes \mathbb{I})|\psi^{AB}\ra}$.

\section{Upper bound for RCC and its relation with Entanglement}

Theorem \ref{th2} shows the necessary and sufficient condition for quantum operations that create nonzero coherence.
Now we study how much RCC can be created and the relation between RCC and the entanglement between A and B.

\begin{lemma}\label{le1}
Under a  quantum operation $\$(\cdot)=\sum_n F_n(\cdot)F^{\dagger}_n$, the remote created coherence for a pure bipartite state $|\psi^{AB}\ra$  is bounded by
\begin{equation}\label{bound1}
C({\rho^A}')\leq \frac{E(|\psi^{AB}\ra)}{p'} \sqrt{\sum_{j<i}|N_{ji}|^2},
\end{equation}
where $C$ is the $l_1$ norm coherence, $E$ is the entanglement measure, concurrence, $p'$ is the probability of getting the state ${\rho^A}'$,
$N=\sum_n {F}^{\dagger}_n F_n\leq \mathbb{I}$ and $N_{ji}$ is the matrix elements under $|\psi^{AB}\ra$'s  Schmidt decomposition's basis of system B.
\end{lemma}

\begin{proof}
From  equations (\ref{cl1}) and (\ref{purerhoap}) we have,
\begin{equation}\label{crhoap}
C({\rho^A}')=1/p' \sum_{i\neq j}\sqrt{\omega_i\omega_j} |N_{ji}|.
\end{equation}
By Cauchy inequality $\sum_i a_i b_i\leq \sqrt{\sum_i {a_i}^2 \sum_i {b_i}^2}$, we  have
\begin{equation}\label{neq1}
\begin{aligned}
\sum_{i\neq j} \sqrt{\omega_i \omega_j} |N_{ji}|
\leq& \sqrt{\sum_{i\neq j} \omega_i \omega_j \sum_{i\neq j} |N_{ji}|^2}\\
=&\sqrt{2 \sum_{i\neq j} \omega_i \omega_j} \sqrt{ \sum_{i< j} |N_{ji}|^2}.
\end{aligned}
\end{equation}
On the other hand, the concurrence of the state (\ref{psiab}) is given by
$E(|\psi^{AB}\ra)=\sqrt{2(1-\tr [(\tr_B |\psi\ra ^{AB} \la \psi |)^2])}=\sqrt{2 \sum_{i\neq j} \omega_i \omega_j}$,
which together with  (\ref{crhoap}) and (\ref{neq1}) complete the proof.
\end{proof}

Under the operation $\$$ acting on the system B, one interesting thing is that the upper bound of RCC (\ref{bound1})
is proportional to the entanglement while inversely proportional to the probability of getting the final state.
Here, one should note that $N_{ji}$s depend on the local Schmidt decomposition basis, hence the bound (\ref{bound1}) not only
depends on the entanglement between A and B, but also the local basis of the system B, as
the coherence is a basis dependent quantity. Next, we investigate
the relations among the average RCC under an operation on the system B
for maximally entangled states, that for an arbitrarily given pure quantum state and the entanglement of this given state.

Consider a pure bipartite entangled state $|\psi^{AB}\ra$
and a trace preserving quantum channel $\$(\cdot)=\sum_n F_n(\cdot) {F_n}^\dagger$,
$\sum_n {F_n}^\dagger F_n =\mathbb{I}$. Under the channel $\$$ on the system B,
the state $|\psi^{AB}\ra$ becomes $(\mathbb{I}\otimes F_n )|\psi\ra ^{AB} \la \psi|(\mathbb{I}\otimes F_n)^\dagger/p_n'$
with probability $p_n'=\tr[(\mathbb{I}\otimes F_n )|\psi\ra ^{AB} \la \psi|(\mathbb{I}\otimes F_n)^\dagger]$.
Bob communicates each outcome to Alice such that, under the channel $\$$, Alice can gain an average coherence over all the outcomes,
$\overline{C}^{A|B}(|\psi^{AB}\ra)=\sum_n p_n' C(\tr_B[(\mathbb{I}\otimes F_n) |\psi\ra ^{AB} \la \psi|(\mathbb{I}\otimes F_n)^\dagger]/p_n')$.

\begin{theorem}\label{th3}
For a pure bipartite entangled state $|\psi^{AB}\ra$ with zero coherence of system A,
under a trace preserving quantum channel $\$(\cdot)=\sum_n F_n(\cdot) {F_n}^\dagger$ on system B,
the average created coherence of system A satisfies the following relation:
\begin{equation}\label{bound2}
\overline{C}^{A|B}(|\psi^{AB}\ra)\leq \frac{d}{2}E(|\psi^{AB}\ra) \overline{C}^{A|B}(|\phi^{AB}\ra),
\end{equation}
where $|\phi^{AB}\ra$ is the maximal entangled state in the Schmidt decomposition  basis of $|\psi^{AB}\ra$,
$\overline{C}^{A|B}(|\psi^{AB}\ra)$ and $\overline{C}^{A|B}(|\phi^{AB}\ra)$ are
the average coherence of systems A under the channel $\$$ for states $|\psi^{AB}\ra$ and $|\phi^{AB}\ra$ respectively,
$C$ is the $l_1$ norm coherence, $d$ is the dimension of system A and $E$ is the concurrence.
\end{theorem}

\begin{proof}
By setting $\omega_i=1/d$, $\forall\, i$, in (\ref{psiab}) we get the maximally entangled state with respect to $|\psi^{AB}\ra$'s Schmidt basis,
\begin{equation}\label{beta}
|\phi^{AB}\ra=1/\sqrt{d} \sum_i |i\ra |\beta_i\ra,
\end{equation}
where $d$ is the dimension of system A. Under the channel
$\$$ on system B, the final states of system A corresponding to $|\psi^{AB}\ra$ and $|\phi^{AB}\ra$ are given by
$$
{\rho^A_n}'=\tr_B[(\mathbb{I}\otimes F_n) |\psi\ra ^{AB} \la \psi|(\mathbb{I}\otimes F_n)^\dagger]/p_n'
$$
and
$$
{\rho^A_n}''=\tr_B[(\mathbb{I}\otimes F_n) |\phi\ra ^{AB} \la \phi|(\mathbb{I}\otimes F_n)^\dagger]/p_n'',
$$
where $p_n'=\tr[(\mathbb{I}\otimes F_n )|\psi\ra ^{AB} \la \psi|(\mathbb{I}\otimes F_n)^\dagger]$ and
$p_n''=\tr[(\mathbb{I}\otimes F_n )|\phi\ra ^{AB} \la \phi|(\mathbb{I}\otimes F_n)^\dagger]$ are the probabilities of
getting the states ${\rho^A_n}'$ and ${\rho^A_n}''$ respectively.
Employing equation (\ref{purerhoap}), we get the coherence
$C({\rho^A_n}')= \sum_{i\neq j}\sqrt{\omega_i\omega_j} |N_{ji}|/p_n',$
and $C({\rho^A_n}'')= \frac{1}{d}\sum_{i\neq j}|N_{ji}|/p_n'',$ where $N=\sum_n {F_n}^\dagger F_n$ and $N_{ji}=\la \beta_j|N|\beta_i \ra$.

Utilizing the following relation
\begin{equation}\label{neq2}
\begin{aligned}
\sqrt{ \sum_{i< j} |N_{ji}|^2}
\leq  \sqrt{(\sum_{i< j} |N_{ji}|)^2}
= \frac{1}{2}\sum_{i\neq j} |N_{ji}|
\end{aligned}
\end{equation}
and Lemma \ref{le1}, we get the relation
\begin{equation}\nonumber
p_n'C({\rho^A_n}')\leq \frac{d}{2} E(|\psi^{AB}\ra) p_n'' C({\rho^A_n}'').
\end{equation}
Since $\$$ is trace preserving, the average coherence satisfies
\begin{equation}\nonumber
\sum_n p_n'C({\rho^A_n}')\leq \frac{d}{2} E(|\psi^{AB}\ra) \sum_n p_n'' C({\rho^A_n}''),
\end{equation}
which gives rise to the relation (\ref{bound2}).
\end{proof}

Theorem 3 shows that for a bipartite state $|\psi^{AB}\ra$ with initial zero coherence of system A, under a tracing preserving channel
of system B, the average increasing of the coherence of system A is bounded by the entanglement between systems A and B,
and the average coherence of system A for the maximally entangled state $|\phi^{AB}\ra$ under the same channel for the system B.
In fact, (\ref{bound2}) is also valid for the system going through some non trace-preserving channels with certain probabilities.
Let $\$_k(\cdot)=\sum_n F^k_n(\cdot) {F^k_n}^\dagger$ be a set of quantum operations such that
$\sum_k\sum_n {F^k_n}^\dagger F^k_n=\mathbb{I}$, i.e., each $\$_k$ is not a trace-preserving channel,
but all $\$_k$ together is. If the system B of a given state $|\psi\ra^{AB}$ and the maximally entangled state $|\phi^{AB}\ra$ go through the
operation $\$_k$ with probability $p_k'=\tr[(\mathbb{I}\otimes \$_k )|\psi\ra^{AB}\la \psi|]$ and
$p_k''=\tr[(\mathbb{I}\otimes \$_k )|\phi\ra^{AB}\la \phi|]$, respectively,
then one can prove that the average coherence $\overline{C}^{A|B}(|\psi^{AB}\ra)=\sum_k p_k' C(\tr_B[(\mathbb{I}\otimes \$_k) |\psi\ra ^{AB} \la \psi|]/p_k')$
and $\overline{C}^{A|B}(|\phi^{AB}\ra)=\sum_k p_k'' C(\tr_B[(\mathbb{I}\otimes \$_k) |\phi\ra ^{AB} \la \phi|]/p_k'')$
satisfy the relation (\ref{bound2}). Here, if each $\$_k$ is given by one Kraus operator, $\$_k(\cdot)=F_k(\cdot) {F_k}^\dagger$, then
one recovers the result for a trace preserving channel.
It should be noticed here that for the remote creation of the averaged coherence, the
one-way classical communication is required. Otherwise, one would end up with super-luminal signaling.

{\it Remark} Averaging the $\$_k$s for equation (\ref{bound1}) we can also obtain a tighter bound for the average RCC.
\begin{equation}\label{bound3}
\overline{C}^{A|B}(|\psi^{AB}\ra)\leq E(|\psi^{AB}\ra) \sum_k\sqrt{\sum_{j<i}|N^k_{ji}|^2},
\end{equation}
where $N^k=\sum_n {F_n^k}^\dagger F_n^k$ and $N_{ji}^k=\la \beta_j|N^k|\beta_i \ra$, with
$|\beta_i \ra$ as given in (\ref{beta}).
This inequality provides a tighter bound than (\ref{bound2}) and could be used  to estimate the average
RCC for a pure entangled state and a set of operations.

\begin{figure}
  \centering
  \label{fig1}
\includegraphics[width=7cm]{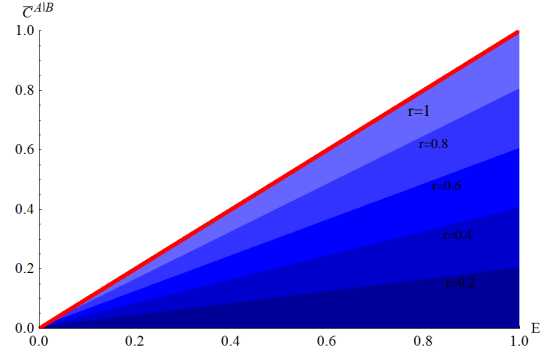}
\caption{
The relation between  entanglement (horizontal ordinate) and the average RCC (vertical ordinate).  We randomly generate over $2\times10^5$  pure states with the form of equation (\ref{psiab}) in $2\otimes 2$ system, of which the average RCCs  are under the phase damping channel $\{F_1=|0\ra \la 0|+\sqrt{1-r}|1\ra \la 1|, F_2=\sqrt{r}|1\ra \la 1|\}$ acting on the system B. The average RCCs, $\overline{C}^{A|B}(|\psi^{AB}\ra)$  (blue dots), get larger as the phase damping rate $r$ increases (the different darkness of blue dots stands for different values of $r$), while the ratio between the average RCC of a state and the average RCC of its corresponding maximal entangled state, $\overline{C}^{A|B}(|\psi^{AB}\ra)/ \overline{C}^{A|B}(|\phi^{AB}\ra)$ (red dots), is always equal to its entanglement for all $r$.}
\end{figure}

Nevertheless, the Theorem \ref{th3} gives a more explicit relation among the entanglement, RCCs of the given state and the corresponding maximally entangled state.
In particular, for two-qubit case, we have the following theorem:

\begin{theorem}\label{th4}
For a pure $2\otimes2$ entangled state $|\psi^{AB}\ra$ with zero coherence of system A,
under a trace preserving quantum channel $\$(\cdot)=\sum_n F_n(\cdot) {F_n}^\dagger$ on system B,
the average RCC of $|\psi^{AB}\ra$ equals to the product of the entanglement of $|\psi^{AB}\ra$ and the average RCC of $|\psi^{AB}\ra$'s corresponding maximally entangled state $|\phi^{AB}\ra$, i.e.,
\begin{equation}\label{eq22}
\overline{C}^{A|B}(|\psi^{AB}\ra)=E(|\psi^{AB}\ra)\, \overline{C}^{A|B}(|\phi^{AB}\ra).
\end{equation}
\end{theorem}

The proof for the above can be derived by noting that, in $2\otimes2$ case, the summation in  (\ref{neq1}) and (\ref{neq2}) only contains one entry which leads to that the equality holds in Theorem \ref{th3}.
The above factorization law manifests that, in a $2\otimes 2$ system, if the operations satisfy the nonzero RCC condition in Theorem \ref{th2}, then the  average RCC of the state is proportional to its entanglement. Since the entanglement is smaller than 1,  its average RCC is  always smaller than  its corresponding maximal entangled state's average RCC. Fig. 1 shows the relations between the average RCC and entanglement for $2\otimes2$ pure states under phase damping channel.

For higher dimensional systems, while the upper bound of the average RCC for a state is proportional to its entanglement, depending on the operations one chooses, its average RCC can exceeds its corresponding maximally entangled state's average RCC. Here it should be noted that the coherence depends on the reference basis, while the entanglement is local unitary invariant. Hence we fix the same basis for the maximally entangled state $|\phi^{AB}\ra$ to the one of $|\psi^{AB}\ra$'s Schmidt basis.

 Besides quantum entanglement, similar relations like (\ref{bound2}) or (\ref{bound3}) may also exist for other quantum correlations like quantum discord, as Theorem \ref{th1} implies that all the non-incoherent-quantum states are useful for RCC, other quantum correlations could be responsible for RCC either.

\section{Conclusions}

In conclusion, we have studied the coherence creation for a system A with zero
initial coherence, with the help of quantum operations on another system B that is correlated to A and one-way classical communication.
We have found that all the non-incoherent quantum states can be used for RCC and all the incoherent-quantum states can not. For pure states, the necessary and sufficient condition of RCC for the quantum operations on system B has been presented. The upper bound of average remote created coherence has been derived, which shows the relation among the entanglement and RCC of the given quantum state, and the RCC of the corresponding maximally entangled state. Moreover, for two-qubit systems, a simple factorization law for the average remote created coherence has been given.

\smallskip
\noindent{\bf Acknowledgments}\, \, We thank B. Chen, Y. K. Wang and S. H. Wang for useful discussions. This work is supported by the NSFC under numbers 11401032, 11275131, 11175094, 11675113 and 91221205, and the National Basic Research Program of China (2015CB921002).

\end{document}